\documentclass[conference]{IEEEtran}

\usepackage{epsfig,graphics,subfigure,amsthm,psfrag,amsmath,amssymb,cite,citesort}
\usepackage{amsmath,amssymb,eucal}
\usepackage{xifthen}
\usepackage{mathtools}
\usepackage{enumerate}
\usepackage{microtype}
\usepackage{xspace}
\usepackage{bm}
\usepackage[T1]{fontenc}
\usepackage{fancyhdr}
\usepackage{lastpage}
\usepackage[center]{caption}
\usepackage{xcolor}
\allowdisplaybreaks

\newcommand{\mbs}[1]{\bm{#1}}
\newcommand{\vect}[1]{{\lowercase{\mbs{#1}}}}
\newcommand{\mat}[1]{{\uppercase{\mbs{#1}}}}

\renewcommand{\Bmatrix}[1]{\begin{bmatrix}#1\end{bmatrix}}
\newcommand{\Pmatrix}[1]{\begin{matrix}#1\end{matrix}}

\newcommand{\T}{{\scriptscriptstyle\mathsf{T}}}
\renewcommand{\H}{{\scriptscriptstyle\mathsf{H}}}

\renewcommand{\Re}[1][]{\ifthenelse{\isempty{#1}}{\operatorname{Re}}{\operatorname{Re}\left(#1\right)}}
\renewcommand{\Im}[1][]{\ifthenelse{\isempty{#1}}{\operatorname{Im}}{\operatorname{Im}\left(#1\right)}}

\newcommand{\ev}{\vect{e}}

\newcommand{\hv}{\vect{h}}

\newcommand{\pv}{\vect{p}}
\newcommand{\qv}{\vect{q}}

\newcommand{\uv}{\vect{u}}
\newcommand{\vv}{\vect{v}}

\newcommand{\xv}{\vect{x}}
\newcommand{\yv}{\vect{y}}

\newcommand{\Am}{\mat{a}}

\newcommand{\Hm}{\mat{h}}

\newcommand{\Km}{\mat{k}}

\newcommand{\Cc}{{\mathcal C}}
\newcommand{\Dc}{{\mathcal D}}

\newcommand{\Uc}{{\mathcal U}}

\newcommand{\CC}{\mathbb{C}}

\newcommand{\CN}[1][]{\ifthenelse{\isempty{#1}}{\mathcal{N}_{\mathbb{C}}}{\mathcal{N}_{\mathbb{C}}\left(#1\right)}}

\renewcommand{\P}[1][]{\ifthenelse{\isempty{#1}}{\mathbb{P}}{\mathbb{P}\left(#1\right)}}
\newcommand{\E}[1][]{\ifthenelse{\isempty{#1}}{\mathbb{E}}{\mathbb{E}\left(#1\right)}}
\renewcommand{\det}[1][]{\ifthenelse{\isempty{#1}}{\mathrm{det}}{\text{det}\left(#1\right)}}
\newcommand{\trace}[1][]{\ifthenelse{\isempty{#1}}{\mathrm{tr}}{\text{tr}\left(#1\right)}}
\newcommand{\rank}[1][]{\ifthenelse{\isempty{#1}}{\mathrm{rank}}{\text{rank}\left(#1\right)}}
\newcommand{\diag}[1][]{\ifthenelse{\isempty{#1}}{\mathrm{diag}}{\text{diag}\left(#1\right)}}

\DeclarePairedDelimiter\Abs{\lvert}{\rvert^2}
\DeclarePairedDelimiter\norm{\lVert}{\rVert}
\DeclarePairedDelimiter\Norm{\lVert}{\rVert^2}

\newcommand{\defeq}{\triangleq}

\newcommand{\nn}{\nonumber}
\newcommand{\fb}{\mathrm{fb}}
\newcommand{\bh}{\mathrm{bh}}

\newcommand{\DoF}{\mathrm{DoF}}

\newtheorem{proposition}{Proposition}

\newtheorem{definition}{Definition}
\newtheorem{theorem}{Theorem}
\newtheorem{lemma}{Lemma}

\begin{document}
{\fontsize{10.5pt}{12.5pt}\selectfont

\title{The DoF of Network MIMO with Backhaul Delays}
\author{\IEEEauthorblockN{Xinping Yi, Paul de Kerret, and David Gesbert}

\IEEEauthorblockA{Mobile Communications Department \\
    EURECOM\\
06560, Sophia Antipolis, France\\
Email: \{xinping.yi, paul.dekerret, david.gesbert\}@eurecom.fr}}

\maketitle

\begin{abstract}
We consider the problem of downlink precoding for Network (multi-cell) MIMO networks where Transmitters (TXs)
are provided with imperfect Channel State Information (CSI). Specifically, each TX receives a delayed channel estimate with the delay being specific to {\em each channel component}. This model is particularly adapted to the scenarios where a user feeds back its CSI to its serving base only as it is envisioned in future LTE networks. We analyze the impact of the delay during the backhaul-based CSI exchange on the rate performance achieved by Network MIMO. We highlight how delay can dramatically degrade system performance if existing
precoding methods are to be used. We propose an alternative robust beamforming strategy which achieves the maximal performance, in DoF sense. We verify by simulations that the theoretical DoF improvement translates into a performance increase at finite Signal-to-Noise Ratio (SNR) as well.
\end{abstract}
\section{Introduction}
The recent years have witnessed considerable work and progress related to the use of multiple-antenna strategies in interference-limited wireless networks. In particular multiple-antenna schemes have proved quite powerful when combined with some form of cooperation or coordination across interfering devices \cite{Gesbert:2010,Karakayali:2006}. Nevertheless, some cost must be paid for extracting cooperation gains, in the form of information exchange, where such information can be CSI or user data related. In the case the network's backhaul (or specific privacy regulations) do not support the sharing of user data, a so-called interference channel arises whereby cooperative beamforming strategies can be implemented provided shared CSI is made available at the TXs. The strategies differ whether interference canceling is available at the user terminals (see e.g., spatial interference alignment \cite{Maddah-Ali:08,Cadambe:2008}) or is not \cite{Jorswieck:2008}. Substantial gains can be offered in the case user data sharing is allowed among transmitters (so-called Network or multi-cell MIMO), in particular in the sense that interference avoidance at transmitter alone is made possible with a reduced overall number of antennas \cite{Gesbert:2010}. Nevertheless real-time CSI sharing remains an important practical challenge for such systems. Previous work on limited CSI feedback model include the case of finite quantizing \cite{Love:08,Bhagavatula:11} and, recently, delayed feedback.

In \cite{Maddah-Ali:12}, it was shown that even completely stale channel feedback, referred to as ``delayed CSIT'', could be used to achieve a larger degrees of freedom (DoF), or pre-log factor. This is achieved by using a novel space-time interference alignment technique, referred to as ``MAT alignment''. One feature of this scheme is that the transmitter treats the delayed feedback as it is fully uncorrelated with the current channel (i.e., worst case scenario). This work was later extended to take into account the possible temporal correlation in the fading channel \cite{Yang:2012,Gou:2012}, resulting in a scheme (referred to in the following as the ``$\alpha$-MAT alignment'') bridging between conventional zero-forcing and the MAT algorithm (when delays exceed the channel's coherence period). All these works apply to uniform feedback delay settings, i.e. whereby all channel components are delayed by the same amount.

In the context of Network MIMO, as specified by current 4G standards, the situation is different as the standards impose that a user feeds back directly to its serving base only, while any further inter-cell CSI exchange must take place over some specific backhaul signaling channels. In reality, signaling over the backhaul introduces additional delays which further degrade the CSI reliability. Thus any CSIT pertaining to an interfering user is subject to a larger delay than that of the served user.
The analysis of Network MIMO with heterogeneous CSI delays does not follow from the homogeneous broadcast setting of \cite{Maddah-Ali:12,Yang:2012,Maleki:12,Xu:12} in any straightforward manner. Instead, a novel specific study of this problem is addressed in this paper.

More specifically, our contributions are as follows:
\begin{itemize}
  \item We adapt the~$\alpha$-MAT alignment developed in \cite{Yang:2012,Gou:2012} to this heterogeneous delay scenario, and show that the DoF achieved is then limited by the {worst} delay after which a channel estimate is obtained at the TX.
  \item Adapting a modified zero-forcing beamforming (ZFBF) scheme \cite{Kerret:2012}, we overcome the CSI discrepancy created by the backhaul delay and propose a DoF-optimal beamforming scheme for the specific two-cell setting.
\end{itemize}

\textbf{Notation}: Matrices and vectors are represented as uppercase and lowercase letters, and matrix transport and Hermitian transport are denoted by $\Am^\T$, $\Am^\H$, respectively. $\hv^{\bot}$ is the normalized orthogonal component of any nonzero vector $\hv$. The approximation $f(P) \sim g(P)$ is in the sense of $\lim_{P \to \infty} \frac{f(P)}{g(P)}=C$, where $C$ is a constant that does not scale as $P$.

\section{System Model}
\subsection{Channel Model}
We consider a two-cell network, where in each cell the transmitter (TX) with single-antenna serves one single-antenna receiver (RX). Both TXs share the user's data symbols and transmit jointly to the two RXs in a BC mode. {The RXs are assumed to have perfect instantaneous CSI relative to the multiuser channel.} The discrete time baseband signal received at RX-$j$ is given by
\begin{align}
y_j(t) &= \hv_{j}^\H(t) \xv(t) + z_j(t)
\end{align}
for any time instant $t$, where $\hv_{j}^\H(t)=[h_{j1}(t) \ h_{j2}(t)]$ is the concatenated channel vector from the two TXs to RX-$j$, $z_j(t) \sim \CN[0,1]$ is the normalized additive white Gaussian noise at RX-$j$, $\xv(t)=[x_1(t) \ x_2(t)]^\T$ is the input signal, where $x_i$ is sent from TX-$i$ and subject to the power constraint $\E \Norm{\xv(t)} \le P$, $\forall\,t$.

We assume the channel to be temporally-correlated. Under the first-order Gauss-Markov model, the channel evolves as~\cite{Tse:05}
\begin{align} \label{eq:G-M}
  \hv_j(t) = \rho \hv_j(t-\tau) - \sqrt{1-\rho^2} \ev(t)
\end{align}
where $\rho \defeq \E [][\hv^\H(t) \hv(t-\tau)] \in [0,1]$ is the channel correlation coefficient, and $\ev(t)$ is a zero-mean unit-variance complex Gaussian process, i.i.d.~across time. To ``fit'' the classical Clarke's isotropic scattering model \cite{Tse:05,Caire:2010}, we let $\rho = J_0(2 \pi f_d \tau)$, where $f_d$ is the Doppler spread, $\tau$ is the time lapse till the prediction, and $J_0(\cdot)$ is the zero-th order Bessel function of the first-kind.

\subsection{CSI Feedback Model}
The CSI feedback is first performed within one cell from RX to its own TX via a feedback link, followed by exchanging CSI between TXs over a backhaul link. We consider a symmetric setting such that there exist two kinds of delays in the CSI flow: the feedback delay, $\tau_{\fb}$, from the RX to its own TX, and the backhaul delay, $\tau_{\bh}$, between TXs. We define $\tau_{jk}$ as the total amount of delays from RX-$j$ to TX-$k$, such that
\begin{align} \label{eq:delays}
  \tau_{jk} = \left\{ \begin{array}{ll} \tau_{\fb}, & j=k \\ \tau_{\fb}+\tau_{\bh}, & j \ne k \end{array} \right..
\end{align}
Further, $\tau_{\fb}$ and $\tau_{\bh}$ are assumed to be known at both TXs. The network model is schematized in Fig.~1.

        \begin{figure}[htb]
        \begin{center}
        \includegraphics[width=0.8\columnwidth]{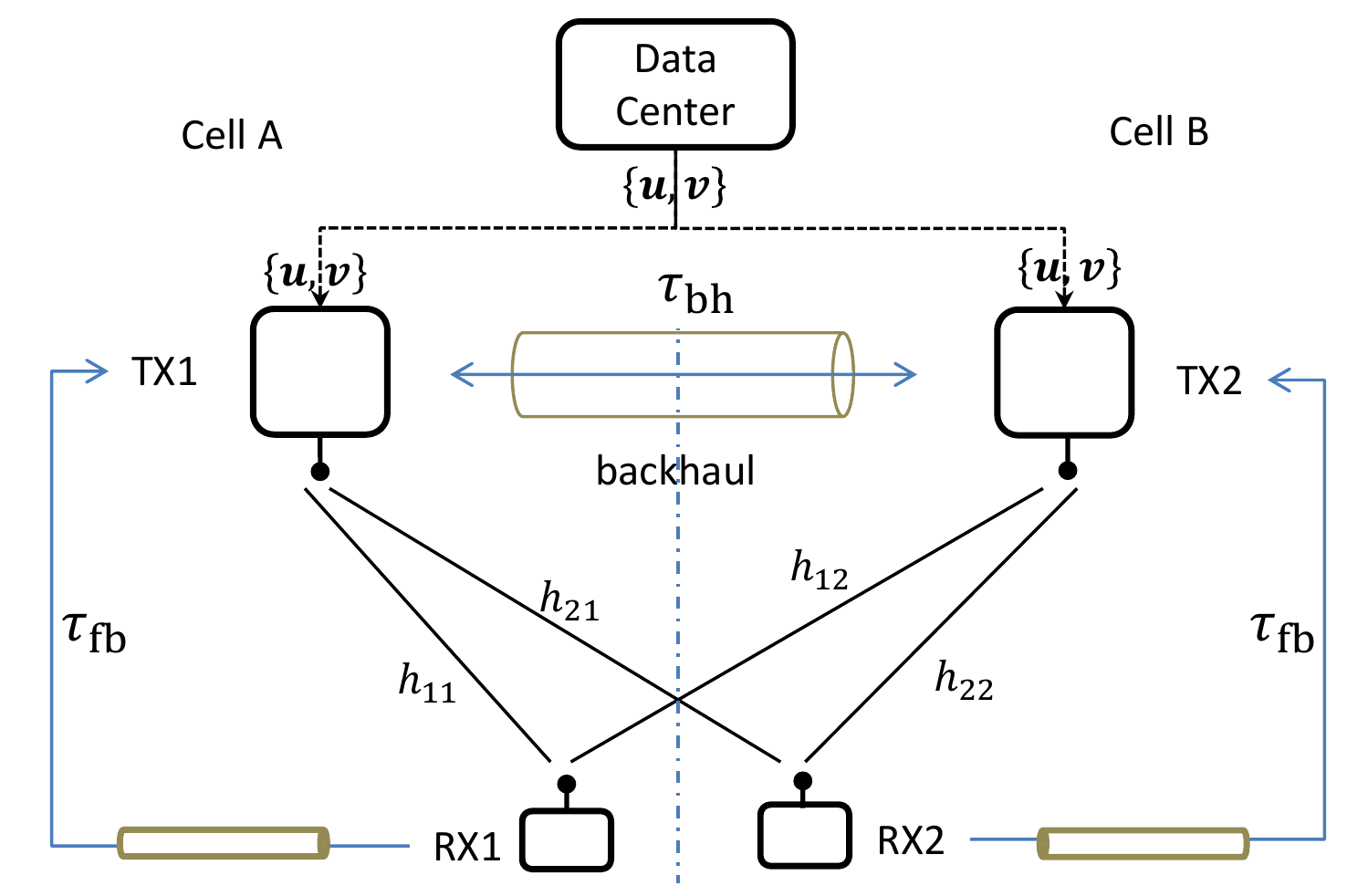}
        \caption{Network model.}
        \label{fig:DoF-1}
        \end{center}
        \vspace{-10pt}
        \end{figure}

At each time instant $t$, we assume TX-$k$ knows perfectly $\hv_{j}$ with delay $\tau_{jk}$, referred to as ``delayed CSIT''. Based on the delayed feedback, TX-$k$ predicts/estimates imperfectly the current channel~$\hv_{j}(t)$, which can be modeled as
  \begin{align} \label{eq:estimate}
    \hv_{j}(t) = \hat{\hv}_{j}^{[k]}(t) + \tilde{\hv}_{j}^{[k]}(t)
  \end{align}
where the estimate $\hat{\hv}_{j}^{[k]}(t)$ and estimation error $\tilde{\hv}_{j}^{[k]}(t)$ are independent and assumed to be zero-mean and with variance $(1-\sigma_{jk}^2)$, $\sigma_{jk}^{2}$, respectively $(0 \le \sigma_{jk}^{2} \le 1)$.

Using \eqref{eq:G-M}, we can then write the estimation error as
\begin{align}
  \sigma_{jk}^2 (\tau_{jk}) = 1-J_0^{2} (2 \pi f_d \tau_{jk}),
\end{align}
where $J_0(\cdot)$ is a monotonic decreasing function of $\tau_{jk}$ in the region of interest such that larger delay $\tau_{jk}$ results in less correlation, and hence a larger estimation error $\sigma_{jk}^2(\tau_{jk})$. We are interested in the design requirements for the feedback and backhaul delays as function of the power $P$. If $\sigma_{jk}^2(\tau_{jk})$ decreases as $1/P$ or faster as $P$ grows, the channel estimate is essentially perfect in terms of DoF and does not lead to any DoF loss even with conventional zero-forcing beamforming (ZFBF)~\cite{Caire:2010}. On the contrary, if $\sigma_{jk}^2(\tau_{jk})$ decreases as or slower than $P^0$, e.g., a constant estimation error, then the DoF collapses to zero \cite{Zhang:11}. Hence, to investigate the impact of the delay on the DoF, we assume that the estimation error $\sigma_{jk}^{2}(\tau_{jk})$ can be parameterized as an exponential function of the power $P$, e.g., $\sigma_{jk}^{2}(\tau_{jk}) \sim P^{-\alpha}$ for some $\alpha \in [0,1]$. More specifically, we introduce a parameter $\alpha_{jk}(\tau_{jk}) \in [0,1]$, such that
\begin{align}
  \alpha_{jk}(\tau_{jk}) \defeq  -\lim_{P \to \infty} \frac {\log \sigma_{jk}^{2}(\tau_{jk})}{\log P}, \label{eq:alpha-def}
\end{align}
to indicate the quality of current CSIT estimate known at high SNR.

In this paper, we consider the symmetric case as in \eqref{eq:delays} for simplicity, where the current CSIT $\hv_j(t)$ is known by TX-$k$, i.e., $\hat{\hv}_{j}^{[k]}(t)$, with quality of
\begin{align}
  \alpha_{jk}(\tau_{jk})=\left\{\Pmatrix{\alpha_1,\\ \alpha_2,}\quad \Pmatrix{\tau_{jk}=\tau_{\fb}~~~~~~~~~~~ \\  \tau_{jk} = \tau_{\fb}+\tau_{\bh}} \right.
\end{align}
where $0 \le \alpha_2 < \alpha_1 \le 1$. We henceforth use $\alpha_1$ and $\alpha_2$ to represent respectively the estimation quality of current CSIT with better (i.e.~less delay) and worse (i.e.~more delay) qualities.

\section{Exploiting Either Imperfect Current or Pure Delayed CSIT}
Following the prediction step described in the previous section, we consider now that each TX has both access to the perfect delayed CSIT and an imperfect estimate of the current CSIT. In the following, we first consider the two known solutions corresponding to the extreme cases where one exploits solely the imperfect current CSIT and the other one utilizes merely the perfect delayed one.

\subsection{ZFBF with Imperfect Current CSIT}
One extreme approach is to perform ZF beamforming by utilizing solely the imperfect current CSIT. As mentioned in the previous section, due to the backhaul delay, different TXs have access to the imperfect current CSIT with different estimation qualities. More specifically, TX-$k$ applies ZFBF only based on its own channel estimates~$\hat{\hv}_{1}^{[k]}(t)$ and~$\hat{\hv}_{2}^{[k]}(t)$, and hence the effective precoder takes the form of
\begin{align} \label{eq:conv-zf}
  \qv_i(t) = \Bmatrix{\left\{{\hat{\hv}_{\bar{i}}^{[1]}(t)}^{\bot}\right\}_1 \\ \left\{{\hat{\hv}_{\bar{i}}^{[2]}(t)}^{\bot}\right\}_2}
\end{align}
where $\bar{i} = i \mod 2 + 1$, and $q_{i1}$ is from the first element of the orthogonal component of ${\hat{\hv}_{\bar{i}}^{[1]}(t)}$, whereas $q_{i2}$ is from the second element of the orthogonal component of ${\hat{\hv}_{\bar{i}}^{[2]}(t)}$. With such beamformer, we have the following proposition.
\begin{proposition}
  For the two-cell Network MIMO with feedback and hackhaul delays, the conventional ZFBF based on the imperfect current CSIT achieves a sum DoF of
  \begin{align} \label{eq:dof-zf}
    \DoF_{\mathrm{ZF}} = 2 \min\{\alpha_1,\alpha_2\} = 2 \alpha_2.
  \end{align}
\end{proposition}
\begin{proof}
  The result is directly obtained from \cite{Kerret:2012}, where $\E \Abs{\hv_i^\H(t) \qv_{\bar{i}}(t)} \sim  P^{-\alpha_2}$, and the residual interference is of power, e.g.,$\E \Abs{\hv_1^\H(t) \qv_{2}(t) v(t)} \sim  P^{1-\alpha_2}$ at RX-1, where $v(t)$ is the data symbol intended to RX-2 and satisfies $\E \Abs{v(t)} \le P$. Hence, the DoF for each RX is $\alpha_2$ for symmetry.
\end{proof}
The estimation error in either link reduces the DoF at both RXs, which is clearly a very detrimental property. Note furthermore that conventional robust precoding such as regularized ZFBF does not lead to any DoF improvement.

\subsection{MAT Alignment with Delayed CSIT}
In the other extreme, it has been recently shown in \cite{Maddah-Ali:12} that it is possible to achieve a larger DoF by exploiting solely the delayed CSIT in single-cell MISO BC by a three-slotted protocol (referred to as ``MAT alignment''). Since the delayed CSIT is equally available at both TXs, the MAT alignment can be applied in our setting with backhaul delay without any modification. For brevity, we describe solely in the following the main steps of a variant of the MAT alignment, which is detailed in \cite{Maddah-Ali:12,Yang:2012}.
\subsubsection{Slot-1} $4$ symbols are sent jointly from the two TXs \emph{without} precoding
\begin{align}
\xv(1)=\uv(1) + \vv(1)
\end{align}
where $\uv(1),\vv(1)\in \CC^{2 \times 1}$ are each made of two user's data symbols intended to RX-$1$ and RX-$2$, respectively, and satisfy $\E \Norm{\uv(1)} = \E \Norm{\vv(1)} \le P$. The interferences $\eta_1 \defeq \hv_1^\H(1) \vv(1)$ and $\eta_2 \defeq\hv_2^\H(1) \uv(1)$ are overheard at RX-$1$ and RX-$2$, respectively.
\subsubsection{Slot-2 \& 3}
In these two slots, the overheard interferences are directly retransmitted by time division, i.e.,
\begin{align}
\xv(2)=\Bmatrix{\eta_1 \\ 0 }, \quad \xv(3)=\Bmatrix{\eta_2 \\ 0 }.
\end{align}
The signal vector received over the three time slots at RX-1 is given by\footnote{The noise term is omitted hereafter for conciseness, since it does not matter in the sense of DoF.}:
\begin{align} \label{eq:channel-model}
\yv_1 = \underbrace{\Bmatrix{\hv_1^\H(1) \\ \mathbf 0 \\ h_{11}^*(3) \hv_2^\H(1) }}_{\text{rank=2}} \uv(1) + \underbrace{\Bmatrix{\hv_1^\H(1)  \\ h_{11}^*(2) \hv_1^\H(1) \\ \mathbf 0}}_{\text{rank=1}} \vv(1).
\end{align}
Note that the interference carrying $\vv(1)$ is aligned in one-dimension, leaving 2-dimensional interference-free subspace for the desired signal $\uv(1)$. Consequently, $\uv(1)$ can be successfully recovered at RX-1. The same rule applies to RX-2 as well. Hence, the total $4$ symbols are delivered over $3$ time slots, yielding a sum DoF of
\begin{align} \label{eq:dof-mat}
\DoF_{\mathrm{MAT}} = \frac{4}{3}.
\end{align}

Albeit fascinating in nature, MAT alignment cannot exploit the channel temporal correlation, no matter how perfect the current CSIT can be predicted from the past ones. This leads intuitively to the question: \emph{Can the imperfect current CSIT be exploited together with the delayed one to gain larger DoF?}

\section{Optimal Use of Both Delayed and Imperfect Current CSIT}
\subsection{The $\alpha$-MAT Alignment}
More recently, a modified MAT alignment developed in \cite{Yang:2012,Gou:2012} has been shown to optimally exploit both the delayed and imperfect current CSIT in terms of DoF in the single-cell two-user MISO BC. The differences with the original MAT alignment are two-fold. First, the imperfect current CSIT is exploited in Slot-1 to minimize the power of overheard interference. Particularly, by balancing the transmitting power of two symbols, with respective power $P$ and $P^{1-\alpha}$, and aligning the symbol with higher power to the null space of the estimated current channel vector, the power of the overheard interference can be reduced from $P$ to $P^{1-\alpha}$. Second, a \emph{compressed/quantized} version of the overheard interferences are retransmitted in Slot-$2$ and Slot-$3$ by time division, allowing for new symbols being superposed on them to get extra DoF. We refer to this scheme as ``{$\alpha$-MAT alignment}'' hereafter.

\subsection{The $\alpha$-MAT Alignment in Network MIMO}
When it comes to the two-cell network, however, such a scheme cannot be directly applied. With heterogeneous CSI delays, the TX in neighboring cell has access to a worse estimate of current CSIT than the TX in home cell. It results in two obstacles to the implementation of the $\alpha$-MAT alignment: (1) how to zero-force the interference, and (2) how to retransmit the residual interference. For the first point, since the TXs do not share the same channel estimate, it is not clear how well the interference can be zero-forced. Regarding the second aspect, the overheard interferences $\eta_k,~k=1,2$ can be reconstructed at any TX, only when the beamformer $\qv_i(t)$ that contains the current channel estimate:
\begin{align} \label{eq:ZF-constraint}
    q_{ik}(t) &= f_i(\{\hat{\hv}_j^{[k]}(\tau),~j=1,2\}_{\tau=1}^{t})
\end{align}
is also available at TX-$k$, where $q_{ik}(t)$ is the $k$-th element of $\qv_i(t)$. It is unfortunately not the case with heterogeneous CSI delays.

{To overcome these two obstacles, we develop here a new version of the $\alpha$-MAT alignment being more robust to the CSI discrepancy in this scenario.}
\subsubsection{Slot-1} With imperfect current CSIT, 4 symbols are sent from two TXs \emph{with} precoding, such that
\begin{align}\label{eq:tx-slot1}
  \xv(1) = [\pv_1(1) \ \qv_1(1)] \uv(1) + [\pv_2(1) \ \qv_2(1)] \vv(1)
\end{align}
where $\uv(1),\vv(1) \in \CC^{2 \times 1}$. The received signals are

{\fontsize{10pt}{12.5pt}\selectfont
\begin{align}
  y_1 (1) &= \hv_1^\H(1) [\pv_1(1) \ \qv_1(1)] \uv(1) + \underbrace{\hv_1^\H(1) [\pv_2(1) \ \qv_2(1)] \vv(1)}_{\eta_1} \nn \\
  y_2 (1) &= \underbrace{\hv_2^\H(1) [\pv_1(1) \ \qv_1(1)] \uv(1)}_{\eta_2} + \hv_2^\H(1) [\pv_2(1) \ \qv_2(1)] \vv(1) \nn
\end{align}
}

where $\pv_i(t),\qv_i(t) \in \CC^{2 \times 1},~i=1,2$ are beamformers. While $\pv_i(t)$ can be taken randomly, $\qv_i(t)$ is designed to be a ZF-type beamformer for minimizing the power of overheard interferences $\eta_k,~k=1,2$. The design of the ZF-type beamformer is a critical design parameter. With TXs having different CSI, this choice becomes non-trivial and is discussed in the following.


\subsubsection*{\underline{Conventional ZFBF}}
One of the possible choice consists in using the conventional ZFBF described in (8). Using the previous results, we know that $\E \Abs{\hv_i^\H(t) \qv_{\bar{i}}(t)} \sim  P^{-\alpha_2}$. The interference terms can then be written as $\eta_1=\eta_{11}+\eta_{12}$ and $\eta_2=\eta_{21}+\eta_{22}$ where
  \begin{align}
  \eta_{11} &= h_{11}^*(1) \Bmatrix{ p_{21}(1) & q_{21}(1)} \vv(1)\label{eq:interference-terms-zf-1}\\
  \eta_{12} &= h_{12}^*(1) \Bmatrix{ p_{22}(1) & q_{22}(1)} \vv(1) \\
  \eta_{21} &= h_{21}^*(1) \Bmatrix{ p_{11}(1) & q_{11}(1)} \uv(1)\\
  \eta_{22} &= h_{22}^*(1) \Bmatrix{ p_{12}(1) & q_{12}(1)} \uv(1) \label{eq:interference-terms-zf-4}.
  \end{align}
Let, for instance, $\eta_1$ and $\eta_2$ be generated by TX-2 and TX-1 respectively, the reconstruction of $\eta_{11}$ at TX-2 and $\eta_{22}$ at TX-1 is a problem.

To solve this problem, we resort to a modification of conventional ZFBF. Note that the TX-$k$ possessing $\hat{\hv}_k^{[k]}(t)$ (c.f.~the estimate with less delay $\tau_{\fb}$), can somehow, e.g., waiting for a moment of $\tau_{\bh}$, estimate a worse version of ${\hv}_k(t)$, i.e., $\hat{\hv}_{k}^{[j]}(t)$,~the estimate also held by TX-$j$, with estimation error of $P^{-\alpha_2}$. Thus, $\hat{\hv}_{k}^{[j]}(t)$ is available at both TXs. Hence, the modified ZFBF can be designed as
\begin{align}
  \bar{\qv}_1(t) = \Bmatrix{\left\{{\hat{\hv}_{2}^{[1]}(t)}^{\bot}\right\}_1 \\ \left\{{\hat{\hv}_{2}^{[1]}(t)}^{\bot}\right\}_2} \quad \bar{\qv}_2(t) = \Bmatrix{\left\{{\hat{\hv}_{1}^{[2]}(t)}^{\bot}\right\}_1 \\ \left\{{\hat{\hv}_{1}^{[2]}(t)}^{\bot}\right\}_2}
\end{align}
making the interference power of $\E \Abs{\eta_k} \sim  P^{1-\alpha_2}${, determined by the worse quality of current CSIT}.

Further, the interference terms in \eqref{eq:interference-terms-zf-1} and \eqref{eq:interference-terms-zf-4} become
\begin{equation}
\begin{aligned}
  \bar{\eta}_{11} &= h_{11}^*(1) \Bmatrix{ p_{21}(1) & \bar{q}_{21}(1)} \vv(1)\\
  \bar{\eta}_{22} &= h_{22}^*(1) \Bmatrix{ p_{12}(1) & \bar{q}_{12}(1)} \uv(1)
\end{aligned}
\end{equation}
which are reconstructible at TX-2 and TX-1, respectively, such that $\eta_k$ ($k=1,2$) can be retransmitted in the ensuing two slots.


\subsubsection*{\underline{Active/Passive-ZFBF}}
To improve over the conventional ZF beamforming, a so-called \emph{Active/Passive ZFBF} (referred to as ``A/P-ZFBF'') developed in \cite{Kerret:2012} for the case of distributed CSIT can be applied here. Differently from the conventional ZFBF, the A/P-ZFBF takes the form of
\begin{align}
  \qv_1(t) = \Bmatrix{1 \\ -\frac{\hat{\hv}_{21}^{[2]*}(t)}{\hat{\hv}_{22}^{[2]*}(t)} } \quad \qv_2(t) = \Bmatrix{ -\frac{\hat{\hv}_{12}^{[1]*}(t)}{\hat{\hv}_{11}^{[1]*}(t)} \\ 1 }
\end{align}
where one element in $\qv_i(t)$ is set independently of the channel realization and the other one is chosen so as to satisfy the orthogonality constraint. Thus, we have
\begin{align}
  &\E \Abs{\hv_1^\H(t) \qv_{2}(t)} \\
  & =\E \left|{\Bmatrix{\hat{h}_{11}^{[1]*}(t)+\tilde{h}_{11}^{[1]*}(t) & \hat{h}_{12}^{[1]*}(t)+\tilde{h}_{12}^{[1]*}(t)} \Bmatrix{ -\frac{\hat{\hv}_{12}^{[1]*}(t)}{\hat{\hv}_{11}^{[1]*}(t)} \\ 1 }} \right|^2 \nn \\
  &=\E \left|{\tilde{h}_{11}^{[1]*}(t) \left( -\frac{\hat{\hv}_{12}^{[1]*}(t)}{\hat{\hv}_{11}^{[1]*}(t)} \right) + \tilde{h}_{12}^{[1]*}(t)}\right|^2\\
  & \sim P^{-\alpha_1}.
\end{align}
Symmetrically, it also holds $\E \Abs{\hv_2^\H(t) \qv_{1}(t)} \sim P^{-\alpha_1}$. Hence, the power of the interference has been reduced to $\E \Abs{\eta_k} \sim  P^{1-\alpha_1}$.

Recall that $\qv_{1}(t)$ only depends on $\hat{\hv}_2^{[{2}]}(t)$ and $\qv_{2}(t)$ only on $\hat{\hv}_1^{[{1}]}(t)$. Thus, the interference overheard by both RXs, which can be written as
\begin{equation}
  \begin{aligned}
    \eta_1 &= \hv_1^\H(1) \qv_{2}(1) \vv(1) \\
    \eta_2 &= \hv_2^\H(1) \qv_{1}(1) \uv(1)
  \end{aligned}
\end{equation}
are reconstructible at TX-1 and TX-2, respectively.

\subsubsection{Slot-2 \& 3}
Instead of forwarding ${\eta}_k~(k=1,2)$ directly, we quantize them to $\hat{\eta}_k~(k=1,2)$ with source codebook size $(1-\alpha)\log P$ bits each, which makes the quantization error negligible provided that the power of ${\eta}_k~(k=1,2)$ is smaller or equal to $P^{1-\alpha}$ \cite{Yang:2012}. Subsequently, we can retransmit their indices, denoted by $c_k~(k=1,2)$, in a digital fashion, expecting less channel resource to be consumed.

The Slot-2 and 3 consist of the broadcasting of the digitalized interferences $\hat{\eta}_1$ and $\hat{\eta}_2$ where in addition to this broadcasting new private symbols are sent with a power such that they do not lead to any additional interference. {Here, we face with the same problem of how to design the ZF-type beamformer as in Slot-1.} For instance with A/P-ZFBF, together with the codewords of the digitalized interferences $c_k~(k=1,2)$ (c.f.~common message) with rate $(1-\alpha_1)\log P$ but power $P$, another precoded fresh symbols (c.f.~private message) are sent with rate $\alpha_1 \log P$ and power scaling as $P^{\alpha_1}$. Then, the common and private messages are encoded by using superposition coding techniques \cite{Cover:2006}, i.e.,
      \begin{equation} \label{slot-2-3}
      ~~~\text{ZFBF:}
      \begin{aligned} ~~\left\{\Pmatrix{
        \xv(2) = \Bmatrix{0 \\ c_1} + \bar{\qv}_{2}(2) u(2) + \bar{\qv}_{1}(2)v(2)  \\
        \xv(3) = \Bmatrix{c_2\\0} + \bar{\qv}_{2}(3) u(3) + \bar{\qv}_{1}(3)v(3)}\right.
      \end{aligned}
      \end{equation}
      \begin{equation} \label{slot-2-3-2}
      \text{A/P-ZFBF:}
      \begin{aligned} \left\{\Pmatrix{
        \xv(2) = \Bmatrix{c_1\\0} + \qv_{2}(2) u(2) + \qv_{1}(2)v(2)  \\
        \xv(3) = \Bmatrix{0 \\ c_2} + \qv_{2}(3) u(3) + \qv_{1}(3)v(3)}\right.
      \end{aligned}
    \end{equation}
where $u(t),v(t)$ ($t=2,3$) are private messages intended to RX-1 and RX-2 respectively, with rate $\alpha_1 \log P$ each and power constraint $\E \Abs{u(t)}=\E \Abs{v(t)} \le P^{\alpha_1}$.

In Slot-2, the received signals are given by (those in Slot-3 can be similarly obtained)
\begin{equation} \nn
  \begin{aligned}
    y_1(2) &= \underbrace{h_{11}^*(2)c_1}_{P} + \underbrace{\hv_1^\H(2) \qv_{1}(2) u(2)}_{P^{\alpha_1}} + \underbrace{\hv_1^\H(2) \qv_{2}(2)v(2)}_{P^0}  \\
    y_2(2) &= \underbrace{h_{21}^*(2)c_1}_{P} + \underbrace{\hv_2^\H(2) \qv_{1}(2) u(2)}_{P^0} + \underbrace{\hv_2^\H(2) \qv_{2}(2)v(2)}_{P^{\alpha_1}}
  \end{aligned}
\end{equation}
Applying successive decoding\footnote{The common message is firstly decoded by treating other lower power interferences as noise, then with its entire term reconstructed and subtracted from the received signal, and finally the private message is subsequently decoded from the remaining signal~\cite{Cover:2006}.}, the common message (c.f.~$c_1$) is first decoded, followed by the private message (c.f.~$u(2)$ for RX-1 and $v(2)$ for RX-2). Thus, a extra DoF of $\alpha_1$ from each symbol $u(2),v(2)$ is yielded. Then, the overheard interference $\eta_k$ can be reconstructed from $c_k$ with the distortion error drown in the noise \cite{Yang:2012}. Note that $\eta_k$ provides not only the interference cancelation for one RX, but also another linearly independent equation for the other RX, making $\uv(1)$ and $\vv(1)$ retrievable, and hence, yielding $2-\alpha_1$ DoF for each RX. As a result, all the data symbols can be decoded such that each RX obtains $2-\alpha_1+2\alpha_1 = 2+\alpha_1$ DoF over three slots. To sum up, we have the following results:
\begin{proposition}
For the two-cell Network MIMO with feedback and hackhaul delays, the conventional ZFBF and A/P-ZFBF with both delayed and imperfect current CSIT achieve sum DoF of, respectively,
\begin{align}
  \DoF_{\alpha\mathrm{-MAT}}^{\mathrm{ZF}} &= \frac{4+2 \min\{\alpha_1, \alpha_2\}}{3} = \frac{4+2 \alpha_2}{3} \label{eq:dof-amat-zf}\\
  \DoF_{\alpha\mathrm{-MAT}}^{\mathrm{A/P-ZF}} &= \frac{4+2 \max\{\alpha_1, \alpha_2\}}{3} = \frac{4+2 \alpha_1}{3}. \label{eq:dof-amat-apzf}
\end{align}
\end{proposition}
It is worth noting that the sum DoF with conventional ZFBF limited by the worst quality of current CSIT (c.f.~longest delay), while that with A/F-ZFBF is solely determined by the least delayed version, meaning that the impact of the backhaul delay could be mitigated\footnote{The A/P-ZFBF is better in the sense of DoF, however, it meets some practical issues, e.g., power unbalance, which are addressed in \cite{Kerret:2012}.}.

In fact, A/P-ZFBF also verifies the following result.
\begin{theorem} [Optimal DoF Region]
The A/P-ZFBF achieves the optimal DoF region for the two-cell Network MIMO with feedback and hackhaul delays, which can be characterized by
\begin{equation}
    \begin{aligned}
      d_1 &\le 1\\
      d_2 &\le 1\\
      2d_1 + d_2 &\le 2+ \max \{\alpha_{1},\alpha_{2}\}\\
      d_1 + 2d_2 &\le 2+ \max \{\alpha_{1},\alpha_{2}\}.
    \end{aligned}
  \end{equation}
\end{theorem}
\begin{proof}
   See Appendix.
\end{proof}

\section{Numerical Results}
In the following, we provide simulations to verify the efficiency of the proposed schemes. We consider a temporally-correlated rayleigh fading channel and average the sum rate in bits/s/Hz over $1000$ random channel realizations. The qualities of imperfect current CSIT are set to be $\alpha_1=1$ and $\alpha_2=0.5$. In Fig.~\ref{fig:DoF}, the sum rate curves of the aforementioned schemes (i.e., conventional ZFBF, MAT-alignment, and the $\alpha$-MAT alignment with conventional ZFBF and A/P-ZFBF) are plotted with regard to SNR. We also provide the sum rate of conventional ZFBF with perfect CSIT for comparison. As shown in the figure, the $\alpha$-MAT alignment with conventional ZFBF achieves a better DoF (slope of sum rate curve at high SNR) than the original MAT alignment but is limited by the estimation quality with the longest delay (i.e., $\alpha_2$), while the A/P-ZFBF achieves the maximal DoF of 2 (i.e., holding the same slope as the perfect CSIT case), since it is solely determined by the best accuracy (i.e., $\alpha_1=1$) of the current CSIT estimate.
\begin{figure}
\centering
\includegraphics[width=1\columnwidth]{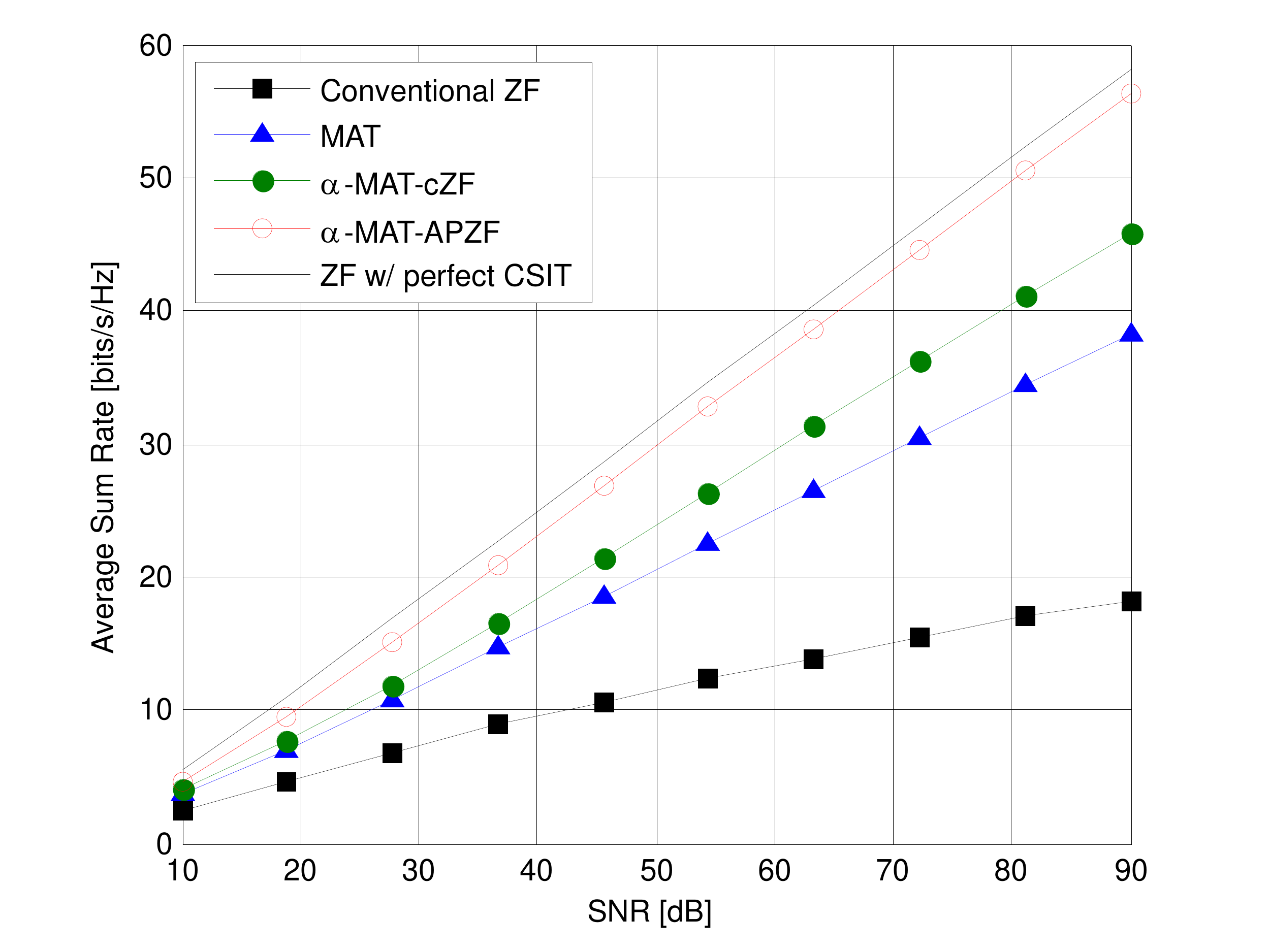}
\caption{Average Sum Rate vs. SNR.}
\label{fig:DoF}
\vspace{-10pt}
\end{figure}

\section{Conclusion}
We address the performance of Network MIMO accounting for the practical constraint that CSI sharing over the backhaul entails additional delays compared with the over-the-air feedback delay. We show that conventional strategies, including recent delayed CSIT-based schemes, fail to maximize DoF. An alternative optimal strategy is given for the specific two-cell case.

\section*{Appendix}
\subsection{Achievability}
The achievability is provided by the $\alpha$-MAT alignment with A/P-ZFBF, where the vertices $(1, \max\{\alpha_1,\alpha_2\})$, $(\max\{\alpha_1,\alpha_2\},1)$ and $\left( \frac{2+\max\{\alpha_1,\alpha_2\}}{3}, \frac{2+\max\{\alpha_1,\alpha_2\}}{3}\right)$ are achievable.

As the vertex $\left( \frac{2+\max\{\alpha_1,\alpha_2\}}{3}, \frac{2+\max\{\alpha_1,\alpha_2\}}{3}\right)$ is demonstrated to be achievable in the previous section, in the following, we present the achievable scheme for the vertex $(1, \max\{\alpha_1,\alpha_2\})$. This vertex can be achieved within one single slot. The transmission with superposition coding can be given by
  \begin{align}
    \xv = \Bmatrix{u_c \\ 0} + \qv_2 u_p + \qv_1 v_p
  \end{align}
  where $u_c$ is a common message and decodable by both RXs but only desirable by RX-1, and $u_p,~v_p$ are private messages which can only be seen and decoded by their corresponding RXs. These transmitted symbols are assumed to satisfy the power constraints $\E\Abs{u_c} \le P$ with rate $(1-\alpha_1)\log P$, and $\E\Abs{u_p} = \E\Abs{v_p} \le P^{\alpha_1}$ with each symbol of rate $\alpha_1 \log P$. At the receiver side, similarly to the procedure in Slot-2 \& 3, both the common (c.f.~$u_c$) and private (c.f.~$u_p,~v_p$) messages can be subsequently decoded by successive decoding, yielding total $1$ and $\max\{\alpha_1,\alpha_2\}$ DoF for Rx-1 and Rx-2, respectively.

  By swapping the roles of two RXs, its counterpart $(\max\{\alpha_1,\alpha_2\},1)$ can be similarly achieved. Hence, with all the vertices achievable, the entire DoF region can be achieved by time sharing.

\subsection{Converse}
Before proceeding further, we present the definition of DoF region.

We denote $\Hm = \{\hv_1, \hv_2\}$ and $\hat{\Hm}^{[k]}(t) \defeq \{\hat{\hv}_{1}^{[k]}(t),\hat{\hv}_{2}^{[k]}(t)\}$.
A $\left(2^{nR_1},2^{nR_2},n\right)$ code scheme is defined to be comprise of:
\begin{itemize}
  \item two message sets, from which two independent messages $W_1$ and $W_2$ intended respectively to the RX-1 and RX-2 are uniformly chosen;
  \item one encoding function at $k$-th BS:
    \begin{align}
    \mathclap{x_k(t) = f_k \left(W_1,W_2, \{\hv_j(\tau)\}_{\tau=1}^{t-\tau_{jk}},\{\hat{\Hm}^{[k]}(\tau)\}_{\tau=1}^{t},~j=1,2\right)} \nn
    \end{align}
    which means the TX can have access to the delayed CSIT of its served RX with less delay and imperfect current CSIT with higher quality;
  \item and one decoding function at $j$-th RX:
    \begin{align}
    \mathclap{\hat{W}_j = g_{j} \left(\{y_j(t)\}_{t=1}^{n},\{\Hm(t)\}_{t=1}^{n},\{\hat{\Hm}^{[k]}(t),~k=1,2\}_{t=1}^{n}\right)} \nn
    \end{align}
\end{itemize}

A rate pair $(R_1,R_2)$ is said to be achievable if there exists a code scheme $\left(2^{nR_1},2^{nR_2},n\right)$, such that the average decoding error probability $P_{e}^{(n)}$, defined as $P_{e}^{(n)} \defeq \P[(W_1, W_2) \neq (\hat{W}_1,\hat{W}_2)]$ vanishes as the code length $n \to \infty$. The capacity region $\Cc$ is defined as the set of all achievable rate pairs. Accordingly, we have the following definition:
\begin{definition} [the DoF region]
  The DoF region for two-cell network MIMO is defined as
  \begin{align}
    \Dc &= \left\{ (d_1,d_2)\in \mathbb{R}_{+}^2 | \forall (w_1,w_2) \in \mathbb{R}_{+}^2, w_1d_1+w_2d_2 \right. \nonumber\\
    &~~~~~~~~~~~~~~~\left.\le \limsup_{P \to \infty} \left( \sup_{(R_1,R_2) \in \Cc} \frac{w_1R_1+w_2R_2}{\log P}\right) \right\}. \nn
  \end{align}
  where $R_i$, $d_i$ are the achievable rate and DoF for the $i$-th RX, respectively.
\end{definition}

To obtain the outer bound, we follow the strategy reminisced in \cite{Yang:2012} to obtain the genie-aided outer bound, by assuming that (a) TX-1 and TX-2 are perfectly cooperated to form a virtual BC, and (c) the RX-2 has the instantaneous knowledge of the RX-1's received signal $y_1(t)$ and also $W_1$. According to the techniques developed in \cite{Yang:2012}, we have the following lemmas:

\begin{lemma} [weighted sum rate bounds \cite{Yang:2012}]
  The weighted sum rate of two RXs for the genie-aided model can be bounded as
\begin{align}
n (2R_1 + R_2)  &\le  2 n \log P + \sum_{t=1}^n (h({y}_1(t),{y}_2(t)|\Uc(t),\Hm(t)) \nn\\
&- 2 h( {y}_1(t)|\Uc(t),\Hm(t)) )+ n \cdot O(1) + n \epsilon_n  \nn
\end{align}
where
\begin{align}
&\Uc(t) \defeq \left\{ \{y_1(k)\}_{k=1}^{t-1},\{y_2(k)\}_{k=1}^{t-1}, \right. \nn \\
&~~~~~~~~~\left.\{\Hm(\tau)\}_{\tau=1}^{t-1} ,\{\hat{\Hm}^{[1]}(\tau)\}_{\tau=1}^{t}, \{\hat{\Hm}^{[2]}(\tau)\}_{\tau=1}^{t}, W_1 \right\}. \nn
\end{align}
\end{lemma}

\begin{lemma} [extremal inequality \cite{Yang:2012}]
  The weighted difference of two differential entropies above can be further bounded as
  \begin{align}
   & h({y}_1(t),{y}_2(t)|\Uc(t),\Hm(t)) - 2 h( {y}_1(t)|\Uc(t),\Hm(t)) \nn \\
  &\le \E_{\hat{\hv}_1^{[1]},\hat{\hv}_1^{[2]},\hat{\hv}_2^{[1]},\hat{\hv}_2^{[2]}} \max_{\Km \succeq 0, \trace (\Km) \le P} \left( \E_{\tilde{\hv}_1^{[1]},\tilde{\hv}_1^{[2]}} \log (1+\hv_1^\H \Km \hv_1) \right. \nn \\
  &~~~~~~~~~~~~~~~~~~~~~~~~~~~~~\left.- \E_{\tilde{\hv}_2^{[1]},\tilde{\hv}_2^{[2]}} \log (1+\hv_2^\H \Km \hv_2) \right) \label{eq:diff-bounds} \\
  & \le \max\{\alpha_1, \alpha_2\} \log P +  O(1)
\end{align}
where $\Km$ is the covariance matrix of channel input.
\end{lemma}
\begin{proof}
  The proof can be straightforwardly extended from \cite{Yang:2012} but with a little modification. The upper bound takes the looser form as in \cite{Yang:2012}, i.e.,
  \begin{align}
  &\E_{\tilde{\hv}_1^{[1]},\tilde{\hv}_1^{[2]}} \log (1+\hv_1^\H \Km \hv_1)  \\
  &\le \log (1+\max\{\norm{\hat{\hv}_1^{[1]}}^2,\norm{\hat{\hv}_1^{[2]}}^2\} \lambda_1) + O(1)
  \end{align}
  while the lower bound can be further loosed as
  \begin{align}
    &\E_{\tilde{\hv}_2^{[1]},\tilde{\hv}_2^{[2]}} \log (1+\hv_2^\H \Km \hv_2)  \\
    &\ge \log (1+2^{\gamma} \min\{\sigma_{kk}^{2},\sigma_{jk}^{2}\} \lambda_1) + O(1)  \\
    & = \log (1+2^{\gamma} \sigma_{kk}^{2} \lambda_1) + O(1)
  \end{align}
  where $j \ne k$, $\sigma_{kk}^{2} \sim P^{-\max \{\alpha_{1},\alpha_{2}\}}$, and $\gamma$ defined in \cite{Yang:2012} is finite. Substituting the above upper and lower bounds into \eqref{eq:diff-bounds}, we obtain
  \begin{align}
    &\E_{\tilde{\hv}_1^{[1]},\tilde{\hv}_1^{[2]}} \log (1+\hv_1^\H \Km \hv_1) - \E_{\tilde{\hv}_2^{[1]},\tilde{\hv}_2^{[2]}} \log (1+\hv_2^\H \Km \hv_2) \nn\\
    & \le \log (1+\max\{\norm{\hat{\hv}_1^{[1]}}^2,\norm{\hat{\hv}_1^{[2]}}^2\} \lambda_1) -  \log (1+2^{\gamma} \sigma_{kk}^{2} \lambda_1) \label{eq:logineq}\\
    &\le \log \left(1+ \frac{\max\{\norm{\hat{\hv}_1^{[1]}}^2,\norm{\hat{\hv}_1^{[2]}}^2\}}{2^{\gamma} \sigma_{kk}^{2}}\right) + O(1)\\
    &\le -\log(\sigma_{kk}^{2}) + O(1) \label{eq:sigmaless1}
  \end{align}
  where \eqref{eq:logineq} follows the inequality $\log\left(\frac{1+a}{1+b}\right) \le \log\left(1+\frac{a}{b}\right)$, $\forall~a,b>0$, and \eqref{eq:sigmaless1} is obtained by noticing $0 \le \sigma_{kk}^{2} \le 1$. This completes the proof.
\end{proof}

Hence, we have
  \begin{align}
  2R_1 + R_2 &\le (2+\max\{\alpha_1, \alpha_2\}) \log P + O(1) + \epsilon_n \nn
  \end{align}
The weighted DoF bound can be obtained according to the definition. It applies to another bound by swapping the role of RX-1 and RX-2.

}
\end{document}